\documentclass[letterpaper, 10 pt, conference]{ieeeconf}  

\IEEEoverridecommandlockouts                              

\let\labelindent\relax
\usepackage{amsmath,amsfonts,amssymb}
\usepackage{array}
\usepackage{subfig}
\usepackage{textcomp}
\usepackage{stfloats}
\usepackage{url}
\usepackage{verbatim}
\usepackage{graphicx}
\usepackage{cite}

\usepackage{cases}
\usepackage{amsthm}
\newtheorem{theorem}{Theorem}
\theoremstyle{definition}
\newtheorem{definition}{Definition}
\theoremstyle{definition}

\theoremstyle{definition}
\newtheorem{lemma}{Lemma}
\theoremstyle{definition}
\newtheorem{problem}{Problem}
\theoremstyle{definition}

\theoremstyle{definition}

\DeclareMathAlphabet{\mathpzc}{OT1}{pzc}{m}{it}
\usepackage{enumitem}
\usepackage{mathrsfs}
\usepackage{algorithm}
\usepackage{algorithmicx}
\usepackage{algpseudocode}
\usepackage{siunitx}
\usepackage{booktabs}
\usepackage{multirow, multicol}
\usepackage{threeparttable}
\usepackage{makecell}
\usepackage{color}

\algrenewcommand\algorithmicrequire{\textbf{Input:}}
\algrenewcommand\algorithmicensure{\textbf{Output:}}

\newlength{\FigureMargin}
\newlength{\FigureMarginBottom}
\newlength{\SubFigureMargin}
\newlength{\SubFigureToCaption}
\newlength{\TableMargin}
\newlength{\TableWithFootnoteMargin}
\newlength{\TableToCaption}
\newlength{\cmidruleWidth}
\newlength{\TheoremMargin}
\newlength{\TheoremMarginBottom}
\newlength{\TheoremWithEquationMarginBottom}
\newlength{\EnumerateMargin}
\title{\LARGE \bf
A Semi-decentralized and Variational-Equilibrium-Based Trajectory Planner for Connected and Autonomous Vehicles
}

\author{Zhengqin Liu$^{1}$, Jinlong Lei$^{2}$, and Peng Yi$^{2}$ 
\thanks{*This research is sponsored by the National Natural Science Foundation of China under Grant 72271187, 62373283, and 62088101, and partially supported by Shanghai Municipal Science and Technology Major Project No2021SHZDZX0100, and the Fundamental Research Funds for the Central Universities, China.}
\thanks{$^{1}$Zhengqin Liu is with the Department of Control Science and Engineering,
Tongji University, Shanghai 201804, China
        {\tt\small (2230709@tongji.edu.cn)}}%
\thanks{$^{2}$ Jinlong Lei and Peng Yi are with the Department of Control Science and Engineering, Tongii University, Shanghai 201804, China; The Shanghai Research Institute for Intelligent Autonomous Systems, Shanghai, 201804, China; Shanghai Institute of Intelligent Science and Technology, Tongji University
       \ {\tt\small (leijinlong,yipeng@tongji.edu.cn)}}%
}

\begin{document}

\maketitle
\thispagestyle{empty}
\pagestyle{empty}

\begin{abstract}
This paper designs a novel trajectory planning approach to resolve the computational efficiency and safety problems in uncoordinated methods by exploiting vehicle-to-everything (V2X) technology. The trajectory planning for connected and autonomous vehicles (CAVs) is formulated as a game with coupled safety constraints. We then define interaction-fair trajectories and prove that they correspond to the variational equilibrium (VE) of this game. We propose a semi-decentralized planner for the vehicles to seek VE-based fair trajectories, which can significantly improve computational efficiency through parallel computing among CAVs and enhance the safety of planned trajectories by ensuring equilibrium concordance among CAVs. Finally, experimental results show the advantages of the approach, including fast computation speed, high scalability, equilibrium concordance, and safety.
\end{abstract}

\section{Introduction}
\label{sec:introduction}
Recently, the interaction mode between autonomous vehicles (AV), i.e., mutual collision avoidance behavior, has received much research attention. Game theory provides a promising approach for modeling the decision-making and interactions of AVs, in which the equilibrium concept can describe each vehicle's optimal action during multi-vehicle interactions. Various equilibrium concepts have been applied to study the trajectory planning of AVs, including Nash equilibrium \cite{meng2016dynamic}, generalized Nash equilibrium (GNE) \cite{yuan2020cooperative}, and Stackelberg equilibrium \cite{zhang2019game}.
However, these studies require the unrealistic assumption that other vehicles' objective functions are known by the ego vehicle. Besides, each vehicle separately solves equilibrium, which leads to redundant computation and discordant equilibrium. The latter results in misjudgments about the actions of other vehicles and damages the trajectories' safety.

With V2X, CAVs can share intentions and decisions, which provides the possibility for the distributed solution of trajectory planning games, thus avoiding redundant computation and ensuring equilibrium consensus among CAVs. Research using distributed frameworks in the games of CAVs seldom focuses on trajectory-level planning. For example, \cite{chen2020joint} uses a cooperative evolutionary game and a distributed algorithm to jointly optimize vehicle routing and traffic signal timing but ignores constraints and dynamics at the trajectory level. Besides, \cite{guo2022lane} employs a decomposed non-cooperative game to decide acceleration/deceleration and steering in lane changing, but it can merely provide CAVs with a set of discrete choices for movement. In addition, \cite{zhang2019algorithm} proposes a distributed trajectory planner for a two-vehicle cooperative game, but does not consider coupled collision avoidance constraints for safety.

To overcome these problems, this paper focuses on traffic scenarios of CAVs that are equipped with roadside units (RSUs). RSUs can communicate with CAVs and coordinate them to avoid collisions. We model CAV's trajectory planning problem as a GNE problem with coupled collision avoidance constraints and propose a semi-decentralized and variational-equilibrium-based planner (SVEP) for CAVs. In the proposed method, each CAV only solves its own trajectory, which ensures real-time performance and scalability. A consensus mechanism is designed by the RSU for Lagrange multipliers to ensure that CAVs converge to the same variational equilibrium (VE), thus enhancing safety. To our knowledge, this is the first study that solves the trajectory planning game of CAVs with coupled constraints in a semi-distributed manner.

Our contributions are as follows.
\begin{enumerate}[wide]
  \vspace{\EnumerateMargin}
  \item We model the trajectory planning problem of CAVs as a game with coupled constraints to enhance the safety of trajectories and respect the autonomy of each vehicle.
  \item The concept of interaction fairness is proposed for the trajectories based on our model. Through sensitivity analysis, we show that an interaction-fair GNE is in fact a VE.
  \item We propose a semi-decentralized and VE-based planner for CAVs, which has the features of fast computation, high scalability, equilibrium concordance, and safety.
  \vspace{\EnumerateMargin}
\end{enumerate}

Notation: $\boldsymbol{0}$ denotes a vector with all elements being $0$. The operator $\operatorname{vec}(\cdot)$ is defined as the concatenation of column vectors or scalars $a_1,\dots,a_l$, i.e., $\operatorname{vec}(a_1,\dots,a_l)=(a_1^T,\dots,a_l^T)^T$.  For a vector-valued function $f(x)$, $J_f(x)$ represents the Jacobian matrix of $f(x)$. For a multivariate function $g(x)$, $\nabla_x g(x)$ denotes the gradient of $g(x)$. For a vector $x$ and a matrix $A$, $\|x\|^2_A=x^TAx$. $A\odot B$ represents the Hadamard product of matrices $A$ and $B$. $\mathcal{N}_{C}$ denotes the normal cone of set $C$, $\operatorname{int} C$ represents the interior point set of set $C$, and $|C|$ denotes the cardinal number of set $C$.

\section{Problem Formulation}
\label{sec:problem_formulation}
\subsection{Traffic Scenario and Trajectory Planning Problem}
\label{subsec:scenario_and_problem}
We consider a two-way traffic scenario equipped with roadside units without traffic signals, where the penetration rate of CAVs is $\qty{100}{\percent}$. The set of CAVs is denoted as $\mathcal{N}=\{1,2,\dots,n\}$. $T$ denotes the discrete prediction horizon, with time steps $k=1,\dots,T$. Denote the state and control vectors of CAV $i$ at time $k$ as $x_i(k)\in \mathbb{R}^{l_x}$ and $u_i(k)\in \mathbb{R}^{l_u}$, respectively. The trajectory to be planned for CAV $i$ is the   sequence $s_i=\operatorname{vec}(s_i(1),s_i(2),\dots,s_i(T))$, where $s_i(k)=\operatorname{vec}(x_i(k),u_i(k))$, but excluding $x_i(1)$ and $u_i(T)$. Fig. \ref{fig:environment_sketch} gives an example of a trajectory planning problem in an intersection scenario.

\begin{figure}[htbp]
              \vspace{\FigureMargin}
              \centering
              \makebox[0.90\linewidth][c]{
              \begin{minipage}{0.57\linewidth}
                \centering
                \subfloat[]{
                  \includegraphics[width=\linewidth]{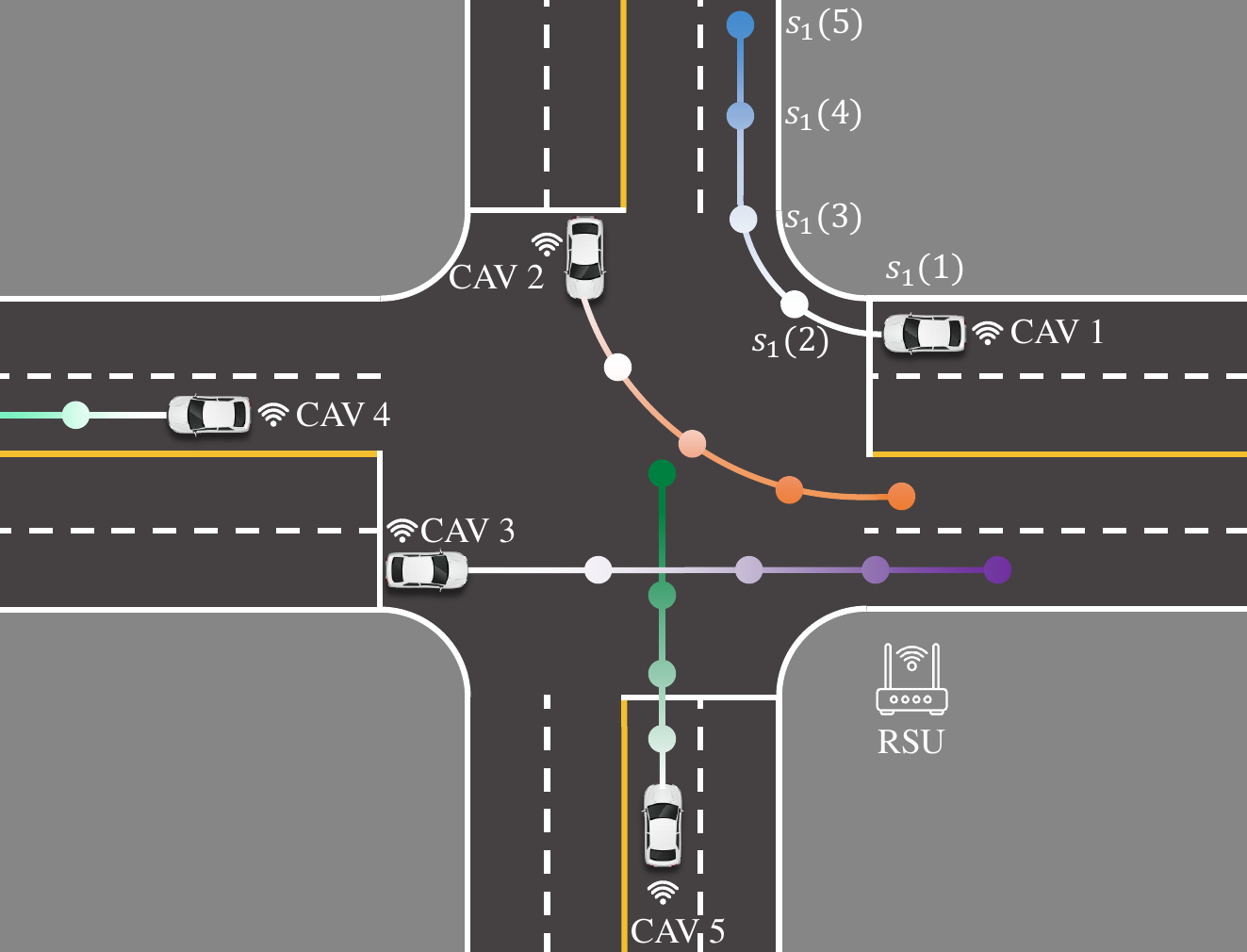}%
                  \label{fig:environment_sketch}
                }
              \end{minipage}
              \begin{minipage}{0.32\linewidth}
                \centering
                \subfloat[]{
                  \includegraphics[width=\linewidth]{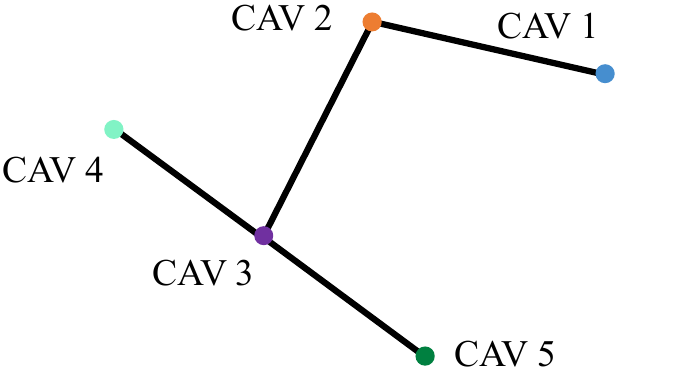}
                  \label{fig:graph_sketch_1}
                }
                \vspace{-0.3cm}
                \newline
                \subfloat[]{
				  \includegraphics[width=\linewidth]{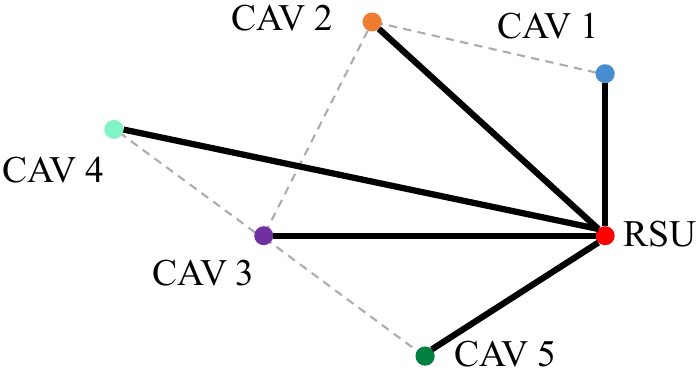}
                  \label{fig:graph_sketch_2}
                }
                \vspace{\SubFigureToCaption}
              \end{minipage}
              }
              \vspace{-0.2cm}
              \caption{The illustration of the problem setting. (a) The traffic scenario and the trajectory planning problem. (b) The interaction graph of CAVs. (c) The communication topology of CAVs and the RSU.}
              \vspace{\FigureMarginBottom}
          \end{figure}

We let the RSU only provide coordination information to CAVs that need to avoid collisions with each other, and CAVs only communicate with the RSU. In the scenario illustrated in Fig. \ref{fig:environment_sketch}, the proximity to other CAVs is used as the criterion for determining the interaction, and the interaction relationship can be represented as an undirected graph $\mathcal{G}=\{\mathcal{N},\mathcal{E}\}$, where $\mathcal{E}$ is the set of edges, as shown in Fig. \ref{fig:graph_sketch_1}. $(i,j)\in\mathcal{E}$ if and only if CAVs $i$ and $j$ have an interaction relationship, and $(i,i)\notin\mathcal{E}$. The neighbor set of CAV $i$ is defined as $\mathcal{N}_i=\{j\in\mathcal{N}|(i,j)\in\mathcal{E}\}$, where $i\notin\mathcal{N}_i$. For any CAV $i \in \mathcal{N}$, its information sharing with any of its neighbors $j\in\mathcal{N}_i$ is through the RSU. Hence, the communication topology is shown in Fig. \ref{fig:graph_sketch_2}.

Let $s_{\mathcal{N}_i}=\operatorname{vec}(s_j),\forall j\in \mathcal{N}_i$ denote the strategy profile of the neighbors of CAV $i \in \mathcal{N}$. The interactive trajectory planning problem  is defined as Problem \ref{de:abstract_game}. Due to the existence of coupled constraints, the problem is a GNE problem (GNEP). We assume each CAV has a feasible trajectory.

\begin{problem}
\vspace{\TheoremMargin}
The trajectory planning GNEP is denoted by a tuple $G=(\mathcal{N},\mathcal{G},(S_i)_{i\in  \mathcal{N}},(X_i)_{i\in  \mathcal{N}},(J_i)_{i\in  \mathcal{N}})$, where $ \mathcal{N}$ is the set of CAVs, $\mathcal{G}$ is the interaction graph of CAVs, $S_i$ and $X_i$ are respectively the private and coupled  strategy constraints for CAV $i$, and $J_i$ is the objective function of CAV $i \in  \mathcal{N}$. Each CAV $i$ obtains the information of its neighbors $j\in \mathcal{N}_i$ and coordination information through the RSU, and solves the minimization problem
\begin{equation}
\min\limits_{s_i\in S_i\cap X_i(s_{\mathcal{N}_i})} J_i(s_i).                      
\label{eq:abstract_game}
\end{equation}
\label{de:abstract_game}
\vspace{\TheoremWithEquationMarginBottom}
\end{problem}

\subsection{Detailed Modeling of Trajectory Planning Problem}
\label{subsec:detailed_model}
This part will define constraints and the objective function so as to give a detailed model of Problem \ref{de:abstract_game}.
          
\begin{enumerate}[wide]
    \item {\bf Dynamics constraints}: We use the kinematic bicycle model \cite{matute2019experimental}. The state of CAV $i$ at time step $k$, denoted by  $x_i(k)=\operatorname{vec}(p_{x,i}(k),p_{y,i}(k),v_{i}(k),\psi_{i}(k)) \in \mathbb{R}^4$, consists of the $x$ and $y$ coordinates, velocity, and yaw angle, while the control $u_i(k)=\operatorname{vec}(a_i(k),\delta_i(k))\in\mathbb{R}^2$ consists of acceleration and front wheel steering angle. The continuous-time dynamics equations are $\dot{p}_{x,i}=v_i\cos\psi_i, \allowbreak \dot{p}_{y,i}=v_i\sin\psi_i, \allowbreak \dot{v}_i=a_i, \allowbreak \dot{\psi}_i=\frac{v_i\tan\delta_i}{L}$, where $L$ is the longitudinal length of a CAV. The discrete-time form of dynamics equations is $\tilde{f}_i(s_i)=\boldsymbol{0}$. We further simplify it through linearization to obtain the dynamics constraints $f_i(s_i)=\tilde{f}_i(\bar{s}_i)+J_{\tilde{f}_i}(\bar{s}_i)\cdot (s_i-\bar{s}_i)$, where $\bar{s}_i$ is the nominal value given by initialization.

    \begin{figure}[htbp]
    \vspace{\SubFigureMargin}
    \centering
    \subfloat[]{
        \includegraphics[width = 0.40\linewidth]{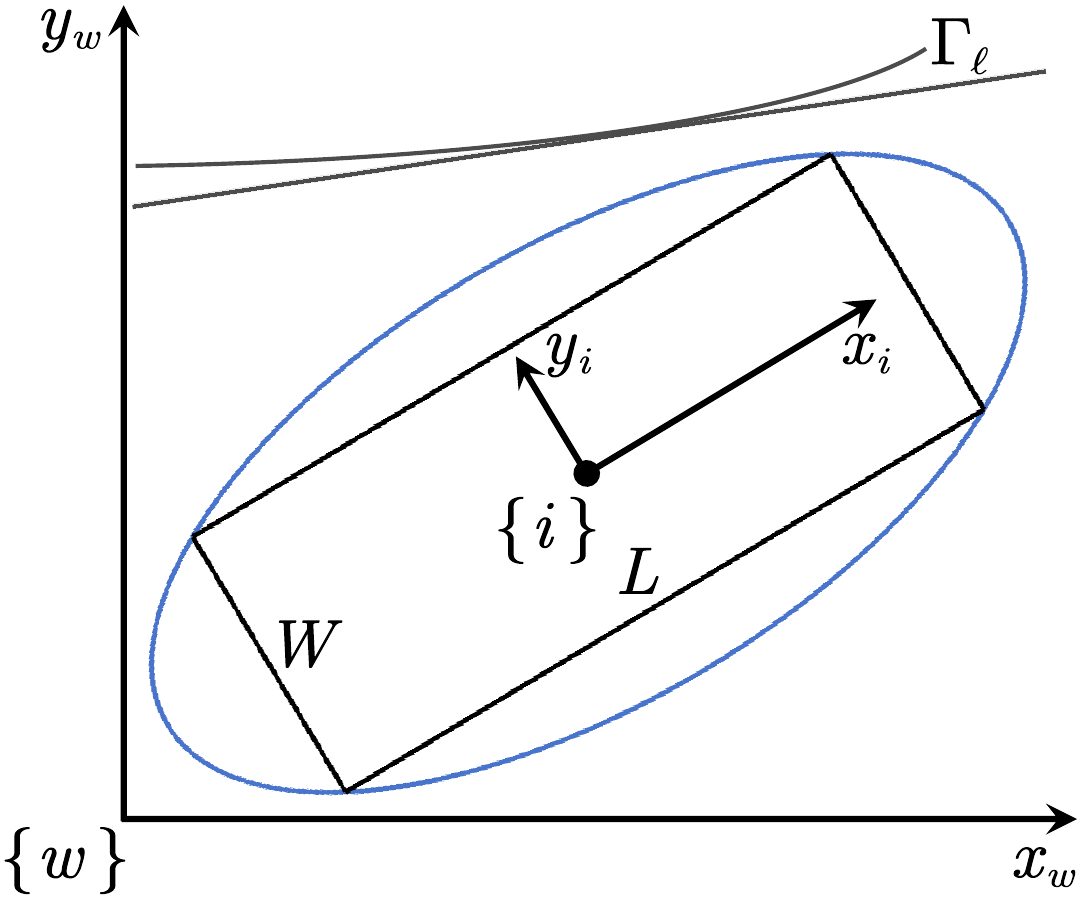}
        \label{fig:lane_constraint}
    }
    \subfloat[]{
        \includegraphics[width = 0.40\linewidth]{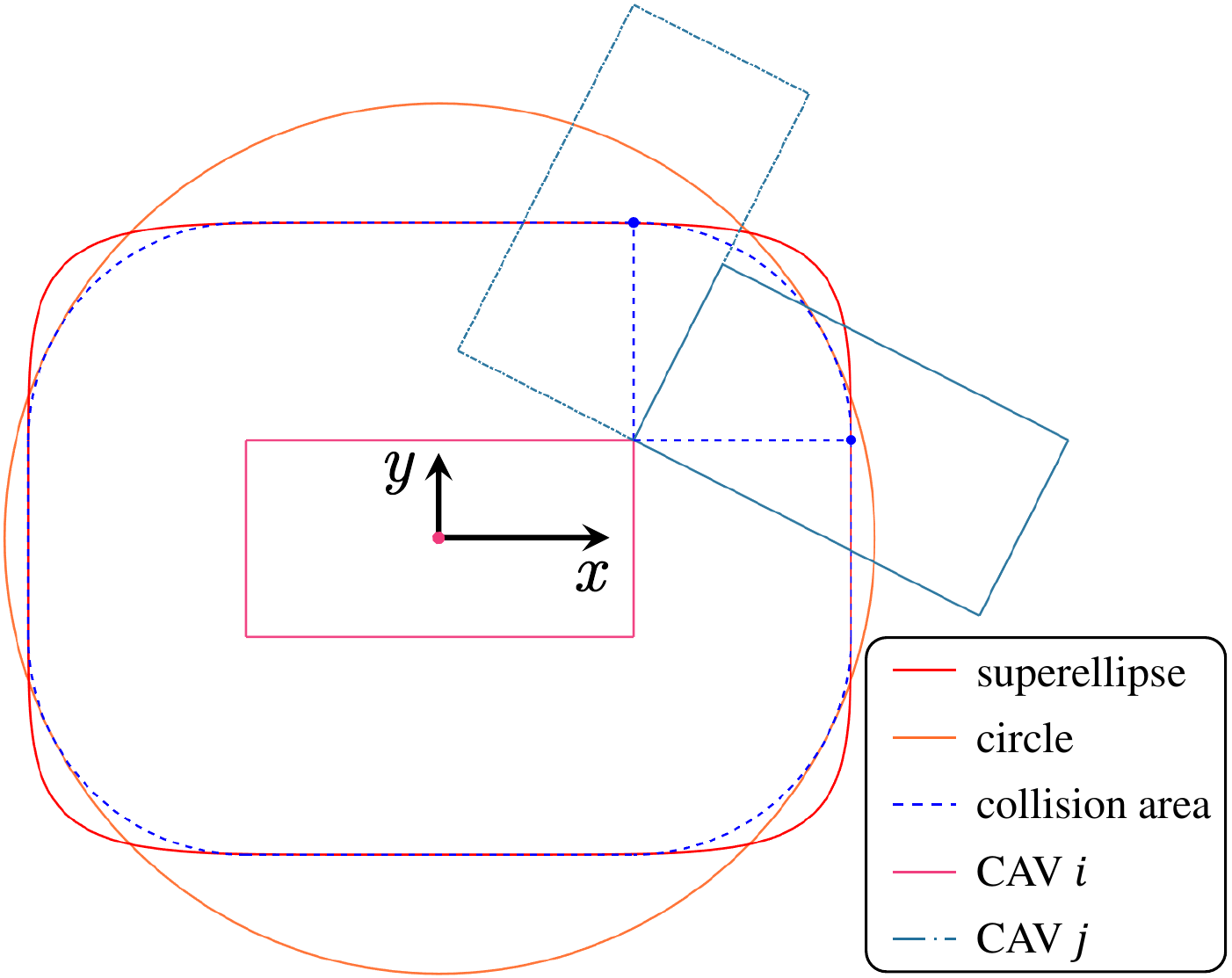}
        \label{fig:collision_avoidance_constraint}
    }
    \vspace{\SubFigureToCaption}
    \caption{Illustration of constraints. (a) Lane constraints. (b) Collision avoidance constraints.}
    \vspace{\FigureMarginBottom}
    \label{fig:constraint}
\end{figure}
       
    \item {\bf Box constraints}: CAV $i$ follows $v_i(k) \allowbreak \in [v_{i,\min}, \allowbreak v_{i,\max}], \allowbreak a_i(k) \allowbreak \in [a_{i,\min}, \allowbreak a_{i,\max}], \allowbreak \delta_i(k) \allowbreak \in [\delta_{i,\min}, \allowbreak \delta_{i,\max}]$.

    \item {\bf Lane constraints}: We use the intersection between the circumscribed ellipse of the vehicle's rectangular plan view and the lane boundary line $\Gamma_\ell$ to determine whether the vehicle is off-lane, as shown in Fig. \ref{fig:lane_constraint}. The equation of the circumscribed ellipse in the vehicle coordinate frame $i$ shown in Fig. \ref{fig:lane_constraint} is given by $\frac{\check{x}^2}{U^2}+\frac{\check{y}^2}{V^2}=1$, and the linearized $\Gamma_\ell$ is $d\check{x}+e\check{y}+f=0$. Substituting the line equation into the ellipse equation yields $\left(d^{2} U^{2}+e^{2} V^{2}\right) \check{x}^{2}+2 d f U^{2} \check{x}+U^{2}\left(f^{2}-V^{2} e^{2}\right)=0$. From the discriminant $\Delta\leq 0$, we obtain the lane boundary constraint $c_{i,\ell}(s_i(k))=d^{2} U^{2}+e^{2} V^{2}-f^{2}\leq 0$. We also restrict the position of CAV $i$ to the inner side of the lane boundary line, i.e., $d_{i,\ell}(s_{i}(k))=ap_{x,i}(k)+bp_{y,i}(k)+c\leq 0$. The dynamic constraints, box constraints, and lane constraints jointly determine $S_i$.

    \item {\bf Collision avoidance constraints}: Let the expression of the superellipse be $\frac{x^6}{(L / 2 + D / 2) ^ 6}  + \frac{y ^ 6}{(W / 2 + D/ 2) ^ 6} = 1$, where $L$, $W$, and $D$ respectively represent the length, width, and diagonal length of the rectangle. This shape can cover the collision area more accurately than the commonly used circle, as shown in Fig. \ref{fig:collision_avoidance_constraint}. At time step $k$, the coordinates of CAV $j$ in the frame $i$ are denoted as $(\check{p}_{x,j}(k),\check{p}_{y,j}(k))$, and the collision avoidance constraint between CAVs $i$ and $j$ is denoted as $h_{i,j}(s_i(k),s_j(k))\leq 0$, where $h_{i,j}(s_i(k),s_j(k))=1-\frac{(\check{p}_{x,j}(k))^6}{(\frac{L}{2} + \frac{D}{2}) ^ 6}  - \frac{(\check{p}_{y,j}(k)) ^ 6}{(\frac{W}{2} + \frac{D}{2}) ^ 6}$. We make the collision avoidance constraints of CAVs $i$ and $j$ take the same form. The linearized collision avoidance constraint between CAV $i$ and all its neighbors is denoted by $h_i(s_i,s_{\mathcal{N}_i})$, and we define $X_{i}\left(s_{\mathcal{N}_i}\right) =\{s_i | h_i(s_i,s_{\mathcal{N}_i})\leq \boldsymbol{0}\}$.

    \item {\bf The objective function}: Given a reference trajectory $x_{ref,i}(k)$, the objective of CAV $i$ is to minimize the deviation of $x_{i}(k)$ from $x_{ref,i}(k)$, as well as to minimize $u_i(k)$. The objective function of CAV $i$ is
          \begin{equation}
              \begin{aligned} J_i(s_i)= & \frac{1}{2}\sum_{k=2}^{T-1} \|x_{i}(k)-x_{ref,i}(k)\|^2_{Q_i}+\frac{1}{2}\sum_{k=1}^{T-1}\|u_{i}(k)\|^2_{R_i} \\
                          & +\frac{1}{2}\|x_i(T)-x_{ref,i}(T)\|^2_{Q_{i,f}},
              \end{aligned}
          \end{equation}
          where $Q_i$, $Q_{i,f}$, and $R_i$ are weight matrices, all being positive definite diagonal matrices.
\end{enumerate}

\section{Solution Concept and Its Fairness}
\label{sec:solution_concept}
The solution to Problem \ref{de:abstract_game} is a GNE $s^{*}=\operatorname{vec}\left(s_{1}^{*}, \dots, s_{n}^{*}\right)$, which satisfies $J_i(s_{i}^{*}) \leq J_i(s_{i}), \forall s_{i} \in S_i \cap X_{i}\left(s_{\mathcal{N}_i}^{*}\right), \forall i \in \mathcal{N}$. Its KKT conditions are given by the following Lemma \ref{le:kkt_set} \cite[Theorem 4.6]{facchinei2010generalized}.

\begin{lemma}
    \vspace{\TheoremMargin}
    If $s^*$ is the GNE of Problem \ref{de:abstract_game}, then $\forall i\in \mathcal{N}$, there exists a Lagrange multiplier vector $\lambda_i^*$, such that
    \begin{equation}
        \begin{array}{c}
        \boldsymbol{0}\in \nabla_{s_i}J_i(s_i^*)+ \nabla_{s_i} h_i(s_i^*,s_{\mathcal{N}_i}^{*})\cdot\lambda_i^*+\mathcal{N}_{S_i}(s_i^*), \\
        h_i(s_i^*,s_{\mathcal{N}_i}^{*})\leq \boldsymbol{0},\ \;
        \lambda_i^*\geq \boldsymbol{0},\ \;
        \lambda_i^* \odot h_i(s_i^*,s_{\mathcal{N}_i}^{*})=\boldsymbol{0},
        \end{array}
        \label{eq:kkt}
    \end{equation}
    where $\mathcal{N}_{S_i}$ represents the normal cone of the set $S_i$, and the gradient operator $\nabla$ follows the denominator layout.
    \label{le:kkt_set}
    \vspace{\TheoremMarginBottom}
\end{lemma}

Since the constraints in Problem \ref{de:abstract_game} are all linear constraints and the cost function is quadratic with a positive definite Hessian matrix, we conclude the convexity of Problem \ref{de:abstract_game} as stated in Lemma \ref{le:convex}.

\begin{lemma}
    \vspace{\TheoremMargin}
    In Problem \ref{de:abstract_game}, for any $ i \in \mathcal{N}$, $J_{i}(s_{i})$ is a convex function;    for any $ i \in \mathcal{N}$ and  any $  s_{\mathcal{N}_i} $, $S_i$ and $X_{i}(s_{\mathcal{N}_i})$ are closed convex sets.
    \label{le:convex}
    \vspace{\TheoremMargin}
\end{lemma}

From the modeling, it is known that the domain of $s_i$ in Problem \ref{de:abstract_game} is $\mathcal{D}=\mathbb{R}^{(l_x+l_u)(T-1)}$. Moreover, as described in Sec. \ref{subsec:scenario_and_problem}, Problem \ref{de:abstract_game} has feasible solutions, which means that there exists an interior point in $\mathcal{D}$ satisfying all constraints. Since all inequality constraints in Problem \ref{de:abstract_game} are affine, according to (5.27) of \cite{boyd2004convex}, Problem \ref{de:abstract_game} satisfies Lemma \ref{le:slater}.

\begin{lemma}
    \vspace{\TheoremMargin}
    (Refined Slater's constraint qualification): In Problem \ref{de:abstract_game}, for any $  i \in \mathcal{N} $ and any $   s_{\mathcal{N}_i} \in \mathbb{R}^{(l_x+l_u)(T-1)|\mathcal{N}_i|},$ there exists  $  s_i \in \operatorname{int}\mathbb{R}^{(l_x+l_u)(T-1)}$ such that $f_i(s_i)=\boldsymbol{0}, b_i(s_i)\leq\boldsymbol{0}, m_i(s_i)\leq\boldsymbol{0}, h_i(s_i,s_{\mathcal{N}_i})\leq\boldsymbol{0}$.
    \label{le:slater}
    \vspace{\TheoremMarginBottom}
\end{lemma}

According to \cite[Chapter 5.2.3]{boyd2004convex}, based on Lemma \ref{le:convex} and \ref{le:slater}, for any given $s_{\mathcal{N}_i}$, strong duality holds in Problem \ref{de:abstract_game}. Under these conditions,  $\lambda_i^*$ satisfying the KKT conditions \eqref{eq:kkt} is the optimal solution to the dual problem. Inspired by \cite[Chapter 5.6.3]{boyd2004convex}, it can be proven that $\lambda_i^*$ is the local sensitivity of $J_i(s_i^*)$ to perturbations $w_i$ in $h_i(s_i,s_{\mathcal{N}_i})$, where $\lambda_i^*=\operatorname{vec}( \allowbreak \lambda_{i,j_1}^*, \allowbreak \dots, \allowbreak \lambda_{i,j_{|\mathcal{N}_i|}}^*), \allowbreak \lambda_{i,j_l}^* = \operatorname{vec}( \allowbreak \lambda_{i,j_l}^*(2), \allowbreak \dots, \allowbreak \lambda_{i,j_l}^*(T) )$, and the same is for $w_i$. $\lambda_{i,j_l}^*(k)$ is the multiplier of $h_{i,j_l}(s_i(k),s_{j_l}(k))$. 

For any given $i\in \mathcal{N}$ and $s_{\mathcal{N}_i} $, the perturbed problem is formulated as
\begin{equation}
    \begin{array}{rl}
        \min\limits_{s_i\in S_i} & J_i(s_i),                   \quad
        \mathrm{s.t.}   ~   h_i(s_i,s_{\mathcal{N}_i})\leq w_i. \\
    \end{array}
    \label{eq:disturb_game}
\end{equation}
Let the optimal value of problem \eqref{eq:disturb_game} be $p(w_i)$. Then it can be proven that $p(w_i)$ and $\lambda_i^*$ satisfy Theorem \ref{th:lambda}.

\begin{theorem} \label{th:lambda}
    Let $\lambda_i^*$ be the multiplier vector satisfying the KKT conditions \eqref{eq:kkt}. Then $\lambda_{i,j}^*=-\frac{\partial p(\boldsymbol{0})}{\partial w_{i,j}}$.
\end{theorem}

\begin{proof}
    Consider the perturbation $\hat{w}_i$, in which $\hat{w}_{i,j}(k)=\alpha\in \allowbreak \mathbb{R}$ and other components are zero. Then, $\lim_{\alpha\to 0} \allowbreak \frac{p(\hat{w}_i)-p(\boldsymbol{0})}{\alpha}= \allowbreak \frac{\partial p(\boldsymbol{0})}{\partial w_{i,j}(k)}$. According to (5.57) of \cite{boyd2004convex}, we have $p(\boldsymbol{0})\leq p(\hat{w}_i)+\alpha\lambda_{i,j}^{*}(k)$. Using this inequality, we obtain
    \begin{equation}
        \frac{p(\hat{w}_i)-p(\boldsymbol{0})}{\alpha}
        \begin{cases}
            \geq -\lambda_{i,j}^{*}(k), & \alpha > 0, \\
            \leq -\lambda_{i,j}^{*}(k), & \alpha < 0.
        \end{cases}
    \end{equation}
    Taking the limit of the above equation yields $\frac{\partial p(\boldsymbol{0})}{\partial w_{i,j}(k)}= \allowbreak \lim_{\alpha\to 0} \allowbreak \frac{p(\hat{w}_i)-p(\boldsymbol{0})}{\alpha}= \allowbreak -\lambda_{i,j}^{*}(k)$. Hence, $\frac{\partial p(\boldsymbol{0})}{\partial w_{i,j}}= \allowbreak -\operatorname{vec}( \allowbreak \lambda_{i,j}^{*}(2), \allowbreak \dots, \allowbreak \lambda_{i,j}^{*}(T)) = \allowbreak -\lambda_{i,j}^*$, i.e., $\lambda_{i,j}^*=-\frac{\partial p(\boldsymbol{0})}{\partial w_{i,j}}$.
    \vspace{\TheoremMarginBottom}
\end{proof}

It is noticed that $h_i(s_i,s_{\mathcal{N}_i})\leq \boldsymbol{0}$ represents the scarcity of collision-free road space resources. $-J_i(s_i)$ can be interpreted as the payoff of CAV $i$. From Theorem \ref{th:lambda} it is seen that $\lambda_{i,j}^*$ equals the marginal rate of increase in the payoff of CAV $i$ when increasing the collision-free road space resource with CAV $j$, which is called the marginal revenue product in economics. Therefore, $\lambda_{i}^*$ directly reflects the value of safety to vehicle $i$. To obtain fair interaction among CAVs, each pair of mutually yielding CAVs should bear the same rate of payoff decrease to avoid collisions. {\it Hence, an interaction-fair GNE is defined when the following is satisfied:}
\begin{equation}
\lambda_{i,j}^*=\lambda_{j,i}^*,\quad     \forall i \in \mathcal{N}, \forall j \in \mathcal{N}_i.
    \label{eq:equal_lambda}
\end{equation}

Next, we explain that the GNE of Problem \ref{de:abstract_game} satisfying \eqref{eq:equal_lambda} is a VE \cite{facchinei2010generalized}. Define the strategy profile of all CAVs as $s=\operatorname{vec}(s_1,\dots,s_n)$, and the pseudo-gradient as $\mathcal{J}(s)=\operatorname{vec}(\nabla_{s_1}J_1(s_1),\dots,\nabla_{s_n}J_n(s_n))$. The definition of VE is given by Definition \ref{de:variational_equilibrium_}.

\begin{definition}
    \vspace{\TheoremMargin}
    For the GNE Problem \ref{de:abstract_game}, if $\forall i\in \mathcal{N}$, $J_{i}\left(s_{i}\right)$ is a convex function, and $\forall i\in \mathcal{N}$, $\forall s_{\mathcal{N}_i} $, $S_i$ and $X_{i}\left(s_{-i}\right)$ are closed convex sets, and there exists a closed convex set $K$ such that $\forall i \in \mathcal{N}$, $X_{i}\left(s_{-i}\right)=\left \{ s_{i}|\left(s_{i}, s_{-i}\right) \in K\right \}$, then the solution to the following variational inequality (VI) problem \eqref{eq:variational_inequality} is also a solution to Problem \ref{de:abstract_game}: find $s^*\in K$ such that
    \begin{equation}
        \langle \mathcal{J}(s),s-s^*\rangle\geq 0,\forall s\in K.
        \label{eq:variational_inequality}
    \end{equation}
    The solution to the VI \eqref{eq:variational_inequality} is called a variational equilibrium.
    \label{de:variational_equilibrium_}
    \vspace{\TheoremMarginBottom}
    \vspace{\TheoremMarginBottom}
    \vspace{\TheoremMarginBottom}
\end{definition}

By Lemma \ref{le:convex}, $S_i$ is a closed convex set, thus $S=\prod_{i=1}^n S_i$ is also a closed convex set. $h_i$ is a linear constraint, so its epigraph is a closed convex set of $s$, and the intersection of epigraphs $X=\{s|h_i(s_i,s_{\mathcal{N}_i})\leq 0,\forall i\in \mathcal{N}\}$ is a closed convex set of $s$. Hence, the closed convex set $K=S\cap X$ satisfies the condition required by Definition \ref{de:variational_equilibrium_}, proving that Problem \ref{de:abstract_game} satisfies the conditions of Definition \ref{de:variational_equilibrium_}.

In the following, we use the KKT conditions to explain that the interaction-fair GNE is a VE. Let each CAV consider the unrepeated concatenation of $h_i(s_i,s_{\mathcal{N}_i})$ as $h(s)\leq \boldsymbol{0}$. Since $h_{i,j}(s_i(k),s_j(k))=h_{j,i}(s_j(k),s_i(k))$, and the collision avoidance constraints of any other two CAVs do not appear in the decision problem \eqref{eq:abstract_game} of CAV $i$, replacing $h_i(s_i,s_{\mathcal{N}_i})\leq \boldsymbol{0}$ with $h(s)\leq \boldsymbol{0}$ does not affect the solution. For such a problem with globally shared coupled constraints, the KKT conditions of VI have been analyzed in \cite{yi2017a}. In $\nabla_{s_i} h(s^*)$, terms unrelated to $s_i^*$ are zero. Thus, $\nabla_{s_i} h(s^*)\cdot\lambda^*=\sum\limits_{j\in\mathcal{N}_i} \nabla_{s_i} h_{i,j}(s_i^*,s_j^*)\cdot \lambda_{i,j}^*$. Then, the KKT conditions of the VI corresponding to the original problem are as shown in Theorem \ref{th:kkt_vi_origin}.

\begin{theorem}
    \vspace{\TheoremMargin}
    If $s^*$ is the solution to the VI \eqref{eq:variational_inequality}, then $\forall i\in \mathcal{N}$, there exists a Lagrange multiplier vector $\lambda_i^*$, such that $\forall i\in \mathcal{N}, \forall j \in \mathcal{N}_i$,
	\begin{align}
			& \boldsymbol{0}\in \nabla_{s_i}J_i(s_i^*)+ \sum\limits_{j\in\mathcal{N}_i} \nabla_{s_i} h_{i,j}(s_i^*,s_j^*)\cdot \lambda_{i,j}^* +\mathcal{N}_{S_i}(s_i^*), \notag \\
			& h_{i,j}(s_i^*,s_j^*)\leq\boldsymbol{0}, \lambda_{i,j}^*\geq \boldsymbol{0}, \lambda_{i,j}^{*T} h_{i,j}(s_i^*,s_j^*)=\boldsymbol{0}, \lambda_{i,j}^*=\lambda_{j,i}^* . \label{eq:kkt_vi_origin}
	\end{align}
	\label{th:kkt_vi_origin}
    \vspace{\TheoremWithEquationMarginBottom}
\end{theorem}

Next, we explain the relationship between the KKT conditions \eqref{eq:kkt} of the GNE and the KKT conditions \eqref{eq:kkt_vi_origin} of the VI \eqref{eq:variational_inequality}. Theorem 4.8 of \cite{facchinei2010generalized} analyzed the case of globally shared coupled constraints. But here the coupled constraints $h_{i,j}(s_i,s_j)$ only exist in the problems of CAV $i$ and $j$, and only CAV $i$ and $j\in\mathcal{N}_i$ have $\lambda_{i,j},\lambda_{j,i}$. Therefore, the condition $\lambda_1=\dots=\lambda_n$ in Theorem 4.8 of \cite{facchinei2010generalized} is equivalent to $\lambda_{i,j}=\lambda_{j,i}, j\in\mathcal{N}_i$ in the context of this work. According to Definition \ref{de:variational_equilibrium_} and the following Theorem \ref{th:kkt_relationship}, we conclude that a GNE of Problem \ref{de:abstract_game} satisfying \eqref{eq:equal_lambda} is a VE.

\begin{theorem}
    \vspace{\TheoremMargin}
    For the GNE Problem \ref{de:abstract_game}, the relationship between the KKT conditions \eqref{eq:kkt} of the GNE and the KKT conditions \eqref{eq:kkt_vi_origin} of the VI \eqref{eq:variational_inequality} is as follows:

    \begin{enumerate}[wide]
        \item If $s^*$ is a solution to the VI \eqref{eq:variational_inequality} such that the KKT conditions \eqref{eq:kkt_vi_origin} hold for $\lambda_{1}^*,\dots,\lambda_{n}^*$. Then $s^*$ is a GNE, and the KKT conditions of the GNE \eqref{eq:kkt} are satisfied with $\lambda_{i,j}=\lambda_{j,i}=\lambda_{i,j}^*=\lambda_{j,i}^*,\forall i\in \mathcal{N}, \forall j\in\mathcal{N}_i$.

        \item If $s^*$ is a GNE, such that the KKT conditions of the GNE \eqref{eq:kkt} are satisfied with $\lambda_{i,j}^*=\lambda_{j,i}^*,\forall i\in \mathcal{N}, \forall j\in\mathcal{N}_i$. Then $(s^*, \lambda_{1}^*,\dots,\lambda_{n}^*)$ is a KKT point of the VI \eqref{eq:variational_inequality} and $s^*$ is a solution to \eqref{eq:variational_inequality}.
    \end{enumerate}
    \label{th:kkt_relationship}
    \vspace{\TheoremMarginBottom}
\end{theorem}

\section{Algorithm}
\label{sec:algorithm}
We propose a semi-decentralized and VE-based planner  as shown in Algorithm \ref{al:framework}, where the superscript $k$ denotes the iteration number and the subscript $1:n$ indicates CAVs $1$ to $n$. The core idea  lies in that each CAV $i\in \mathcal{N}$ and the RSU alternately optimize $s_i$ and $\lambda_i$ to meet the KKT conditions. The core steps of Algorithm \ref{al:framework} are as follows.

\begin{enumerate}[wide]    
    \item \textsc{Minimize\_Augmented\_Lagrangian\_Function} in line \ref{all:minimize_augmented_lagrangian_function}.
    
    (a) Definition of the augmented Lagrangian function. For problem \eqref{eq:abstract_game} of CAV $i$, we use a slack variable $\gamma_{h_i}$ to get a equality constraint: $h_i(s_i,s_{\mathcal{N}_i})+\gamma_{h_i}=\boldsymbol{0}$. Then the augmented Lagrangian function is defined as $L_{i}(s_i,s_{\mathcal{N}_i},\gamma_{h_i},\lambda_i)=J_i(s_i)+\lambda_i^{T}(h_i(s_i,s_{\mathcal{N}_i})+\gamma_{h_i})+\frac{1}{2}\|h_i(s_i,s_{\mathcal{N}_i})+\gamma_{h_i}\|^2_{D_{h_i}}$, 
          where $D_{h_i}$ is a diagonal positive definite penalty matrix.

          (b) The elimination of $\gamma_{h_i}$. When minimizing $L_i$, 
          by fixing $s_i$, we obtain $\min\limits_{\gamma_{h_i}\geq \boldsymbol{0}}\lambda_i^T\gamma_{h_i}+\frac{1}{2}\|h_i(s_i,s_{\mathcal{N}_i})+\gamma_{h_i}\|^2_{D_{h_i}}$. From the optimality conditions of the convex optimization problem with non-negative constraints, we obtain the optimal solution
          \begin{equation}
              \gamma_{h_i}(s_i)=\max\{-D_{h_i}^{-1}\lambda_i-h_i(s_i,s_{\mathcal{N}_i}),\boldsymbol{0}\}.
              \label{eq:optimal_gamma}
          \end{equation}

          (c) Processing and solving the optimization problem. In the $k$-th iteration of the augmented Lagrangian method, given $s_{\mathcal{N}_i}^{k}$, CAV $i$ solves the following problem
          \begin{equation}
              \min\limits_{s_i\in S_i}     L_{i}(s_i,s_{\mathcal{N}_i}^k, \gamma_{h_i}(s_i),\lambda_i^k),
              \label{eq:lagrange_k_problem}
          \end{equation}
          Since \eqref{eq:optimal_gamma} includes $\max$, \eqref{eq:lagrange_k_problem} is a non-smooth function. To convert it to a smooth function, a decision variable $w_{h_i}$ is introduced to replace the term $h_i(s_i,s_{\mathcal{N}_i})+\gamma_{h_i}$, while constraints $-w_{h_i}+h_i(s_i,s_{\mathcal{N}_i})\leq \boldsymbol{0}, -w_{h_i}-D_{h_i}^{-1}\lambda_i\leq \boldsymbol{0}$ are imposed. This approach does not change the optimal solution. After processing, in the $k$-th iteration, the augmented Lagrangian function becomes $L_{i}^k(s_i,w_{h_i},\lambda_i^k)= J_i(s_i)+\lambda_i^{k\ T}w_{h_i}+\frac{1}{2}w_{h_i}^{T} D_{h_i}^k w_{h_i}$, and CAV $i$ solves the following quadratic programming problem
          \begin{equation}
              \begin{array}{rl}
                  \min\limits_{s_i\in S_i, w_{h_i}} & L_{i}^k(s_i,w_{h_i},\lambda_i^k),                         \\
                  \mathrm{s.t.}              & -w_{h_i}+h_i(s_i,s_{\mathcal{N}_i}^k)\leq \boldsymbol{0}, \\
                                             & -w_{h_i}-(D_{h_i}^k)^{-1}\lambda_i^k\leq \boldsymbol{0}.
              \end{array}
              \label{eq:k_iteration_problem}
          \end{equation}

    \item \textsc{Is\_Convergent} in line \ref{all:solution_is_convergent}.
    
    Constraint violation is defined as $\mathpzc{v}( s^{k+1}, \lambda_{1:n}^k, D_{1:n}^k ) = \|\operatorname{vec}( \max \{ h_i(s_i^{k+1}, s_{\mathcal{N}_i}^{k+1} ), -(D_{h_i}^k)^{-1} \lambda_i^k \} )\|_2, \forall i\in \mathcal{N}$. When the constraint violation is below the threshold $\epsilon$, the solution $s^{k+1}$ converges and the iteration is terminated.

\vspace{-0.3cm}
\begin{algorithm}
    \caption{Semi-decentralized and Variational-Equilibrium-Based Trajectory Planner (SVEP)}
    \begin{algorithmic}[1]
        \Require Initial strategy profile $s^0$, penalty update factor $\rho$, maximum iteration times $k_{\max}$
        \Ensure Strategy profile $s$
        \State RSU: $\mathcal{N}_1, \dots, \mathcal{N}_n \leftarrow$ \Call{Determine\_Interaction\_ Relationship}{$s^0$}; \label{all:determine_interaction_relationship}
        \State CAV $i\in \mathcal{N}$: $\lambda_i^0, D_{h_i}^0 \leftarrow$\Call{Initialize\_Parameter}{$\mathcal{N}_i$} \textup{in parallel}; \label{all:initialize_parameter}
        \For{$k = 0$ \textbf{to} $k_{\max}$} \label{all:augmented_lagrange_for_begin}
            \State CAV $i\in \mathcal{N}$: $s_i^{k+1} \leftarrow$\Call{Minimize\_Augmented\_ Lagrangian\_Function}{$s_{\mathcal{N}_i}^k, \lambda_i^k, D_{h_i}^k$} \textup{in parallel}; \label{all:minimize_augmented_lagrangian_function}
            \If{\Call{Is\_Convergent}{$s^{k+1}, \lambda_{1:n}^k, D_{h_{1:n}}^k$}}  \label{all:solution_is_convergent} 
                \State \textbf{break};
            \EndIf \label{all:solution_is_convergent_end}
            \State RSU: $\lambda_1^{k+1}, \dots, \lambda_n^{k+1}\leftarrow$ \textup{\textsc{Update\_Multiplier}} $(s^{k+1}, \lambda_{1:n}^k, D_{h_{1:n}}^k)$; \label{all:update_multiplier} 
            \State CAV $i\in \mathcal{N}$: $D_{h_i}^{k+1}\leftarrow$\Call{Update\_Penalty\_Mat-rix}{$\rho, D_{h_i}^k$} \textup{in parallel}; \label{all:update_penalty_matrix}
        \EndFor \label{all:augmented_lagrange_for_end}
        \State \Return $s^{k+1}$;
    \end{algorithmic}
    \label{al:framework}
\end{algorithm}
\vspace{-0.3cm}

    \item \textsc{Update\_Multiplier} in line \ref{all:update_multiplier}.
    
    For problem \eqref{eq:lagrange_k_problem} in the $k$-th iteration, the optimal solution $s^{k+1}_i$ satisfies
          \begin{equation}
              \begin{aligned}
                  \boldsymbol{0}= & \nabla_{s_i}J_i(s_i^{k+1})+\nabla_{s_i}h_i(s_i^{k+1},s_{\mathcal{N}_i}^k)\times                          \\
                                  & \max\{\lambda_i^k+D_{h_i}^kh_i(s_i^{k+1},s_{\mathcal{N}_i}^k),\boldsymbol{0}\}+\mathcal{N}_{S_i}(s_i^*).
              \end{aligned}
              \label{eq:k_iteration_optimal}
          \end{equation}
          Comparing \eqref{eq:k_iteration_optimal} with the stability condition in the KKT conditions of the slack problem
          \begin{equation}
              \boldsymbol{0}\in \nabla_{s_i}J_i(s_i^*)+ \nabla_{s_i} h_i(s_i^*,s_{\mathcal{N}_i}^{*})\cdot\lambda_i^*+\mathcal{N}_{S_i}(s_i^*),
          \end{equation}
          we obtain the multiplier update formula $\lambda_i^{k+1,md}=\max\{\lambda_i^k+D_{h_i}^kh_i(s_i^{k+1},s_{\mathcal{N}_i}^k),\boldsymbol{0}\}$. To get a VE, \eqref{eq:equal_lambda} must be satisfied. Thus, the multiplier consensus step is designed as
          \begin{equation}
              \lambda_{i,j}^{k+1}=\frac{1}{2}(\lambda_{i,j}^{k+1,md}+\lambda_{j,i}^{k+1,md}),\forall i\in \mathcal{N},\forall j\in\mathcal{N}_i.
          \end{equation}
        Therefore, the multipliers satisfy $\lambda_{i,j}^{k+1}=\lambda_{j,i}^{k+1}$ after each iteration.
    
    \item \textsc{Update\_Penalty\_Matrix} in line \ref{all:update_penalty_matrix}.
    
    At the end of each iteration, we set $D_{h_i}^{k+1}=\rho D_{h_i}^k$.
\end{enumerate}

\section{Experimental Results}
\label{sec:experimental_results}
\subsection{Experimental Setting}
\label{subsec:experimental_setting}
This section conducts experiments in the environment of a two-way, two-lane intersection, as shown in Fig. \ref{fig:experimental_scenario}. 

\begin{figure}[htbp]
    \vspace{\FigureMargin}
    \centerline{\includegraphics[width=0.88\linewidth, keepaspectratio]{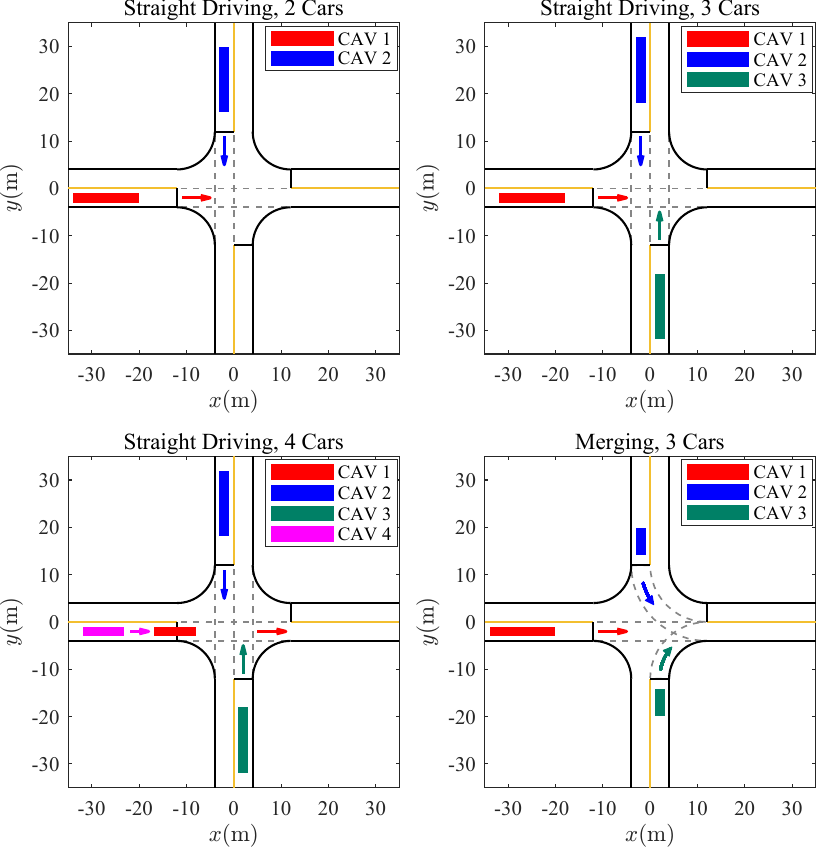}}
    \caption{The experimental scenario and the situation setup. The solid rectangles marked in the legend represent the area covered by the random initial positions of each CAV.}
    \vspace{\FigureMargin}
    \label{fig:experimental_scenario}
\end{figure}

We compare the proposed Algorithm 1 (SVEP) with ALGAMES \cite{le2022algames}. Both of them use the same $T$, $T_s$, and $\epsilon$, while other parameters of ALGAMES follow the default settings in its open-source implementation. The parameters of SVEP are presented in Table \ref{tab:simulation_parameter}. We set $\epsilon=\num{1e-1}$ for ALGAMES in the merging scenario, since otherwise, its other metrics in the merging scenario cannot be analyzed. The program runs on a desktop computer equipped with an Intel Core i5-10400F CPU, 16GB of RAM, and the Windows 11 operating system.

\begin{table}[htbp]
    \vspace{\TableMargin}
    \caption{Algorithm parameters.}
    \vspace{\TableToCaption}
    \begin{center}
        \begin{tabular}{cccc}
            \toprule
            \textbf{Parameter} & \textbf{Value} & \textbf{Parameter} & \textbf{Value} \\
            \midrule
            \shortstack{Maximum\\number of\\iterations $k_{\max}$} & 40 & \shortstack{Discrete\\prediction\\horizon $T$} & $\num{20}$ \\
            \cmidrule[\cmidruleWidth]{1-2} \cmidrule[\cmidruleWidth]{3-4}
            \shortstack{Penalty update\\factor $\rho$} & $\num{4}$ & \shortstack{Discrete\\period $T_s$} & $\qty{0.1}{s}$ \\
            \cmidrule[\cmidruleWidth]{1-2} \cmidrule[\cmidruleWidth]{3-4}
            \shortstack{Constraint violation\\threshold $\epsilon$} & $\num{1e-3}$ & \shortstack{Precision\\threshold $\eta$} & $\num{1e-8}$ \\
            \bottomrule
        \end{tabular}
        \label{tab:simulation_parameter}
    \end{center}
    \vspace{\TableMargin}
\end{table}

\subsection{Evaluation and Analysis}
We designed four situations for the Monte Carlo experiments, as shown in Fig. \ref{fig:experimental_scenario}. The initial speeds of the CAVs are uniformly distributed in $[5, 15]\unit{m/s}$. For $\forall i\in \mathcal{N}$, the value of $D_{h_i}^0$ is set to $D_{h_i}^0 = diag(d_{h_i},\dots, d_{h_i})_{ |\mathcal{N}_i|}$, where $d_{h_i}\sim U[0.5, 1.5]$. Next, we will evaluate the algorithm's performance in terms of computational efficiency and safety.

\begin{enumerate}[wide]
    \item {\bf Computational Efficiency.} We consider the time cost of SVEP, which consists of the computation time per vehicle and the RSU computation time.
          
          (a) Computation time per vehicle: i.e., the average computation time per trajectory planning for each vehicle. The results of 500 Monte Carlo experiments are shown in Fig. \ref{fig:time_comparison}, where  SVEP and ALGAMES use the left and right vertical axes, respectively, to visualize the data within their appropriate ranges.
          It can be seen that SVEP has shorter computation time per vehicle. As the number of straight-driving vehicles increases from 2 to 4, the median of the average computation time of SVEP and ALGAMES increases by $\qty{38.4}{\percent}$ and $\qty{290.8}{\percent}$, respectively. It is evident that SVEP is scalable.

          \begin{figure}[htbp]
              \vspace{\FigureMargin}
              \centerline{\includegraphics[width=0.88\linewidth, keepaspectratio]{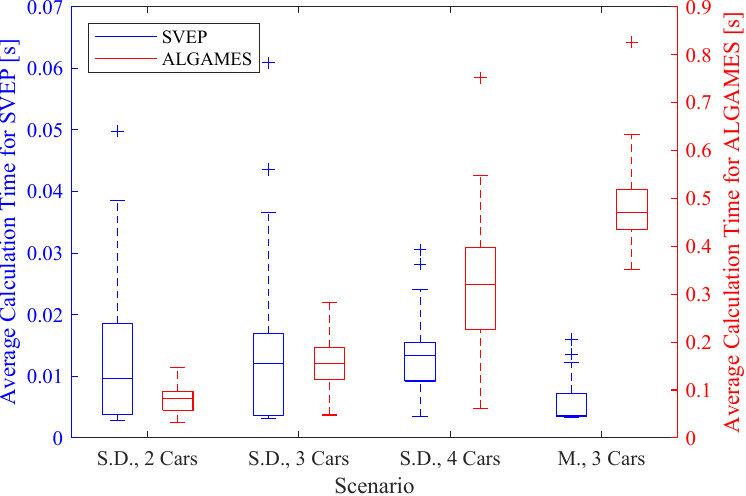}}
              \caption{The average computation time of CAV $1$ in straight-driving (S.D.) and merging (M.) scenarios using SVEP and ALGAMES.}
              \vspace{\FigureMarginBottom}
              \label{fig:time_comparison}
          \end{figure}
          
          (b) The RSU computation time: The mean and standard deviation of the average computation time of the RSU in each trajectory planning are recorded in Table \ref{tab:rsu_time}. It is seen that the average computation time of the RSU is not the main part of SVEP's running time.
            \begin{table}[htbp]
              \vspace{\TableMargin}
              \caption{The average computation time of the RSU.}
              \vspace{\TableToCaption}
              \begin{center}
                  \begin{tabular}{cc}
                      \toprule
                      Scenario & Average Calculation Time [$\unit{s}$] \\
                      \midrule
                      Straight driving, 2 cars & $\num{6.53e-04}\pm\num{4.45e-04}$ \\
                      Straight driving, 3 cars & $\num{1.14e-03}\pm\num{7.10e-04}$ \\
                      Straight driving, 4 cars & $\num{1.96e-03}\pm\num{7.32e-04}$ \\
                      Merging, 3 cars          & $\num{5.63e-04}\pm\num{2.94e-04}$ \\
                      \bottomrule
                  \end{tabular}
                  \label{tab:rsu_time}
              \end{center}
              \vspace{\TableMargin}
          \end{table}
          
          In summary, SVEP's running time in all four experimental situations is below $\qty{1e-1}{s}$, which meets the requirements for real-time planning.

    \item {\bf Safety.}
          The success rate is defined as the proportion of experiments in which all CAVs reach their designated positions without violating any constraints. 500 Monte Carlo experiments were conducted with the same settings, and the success rates are shown in Table \ref{tab:success_rate}. 
          It is seen from Table \ref{tab:success_rate} that SVEP can maintain high safety in every situation without violating constraints. In particular, the safety advantages of SVEP are significant in the merging scenario. Taking the average initial state of the merging scenario as an example, the driving trajectories and front wheel steering angles are shown in Fig. \ref{fig:trajectory_i1s1l1rm}, which indicates that SVEP avoids unnecessary back-and-forth steering, reducing the risk of exceeding the lane boundaries.

          \begin{table}[htbp]
              \vspace{\TableMargin}
              \caption{Success rate.}
              \vspace{\TableToCaption}
              \begin{center}
                  \begin{tabular}{ccc}
                      \toprule
                      \multirow{2}{*}[-0.5\dimexpr \aboverulesep + \belowrulesep + \cmidrulewidth]{Scenario} & \multicolumn{2}{c}{Algorithm} \\ 
                      \cmidrule{2-3}
                      & SVEP & ALGAMES \\ 
                      \midrule
                      Straight driving, 2 cars & $\qty{100}{\percent}$& $\qty{98.8}{\percent}$ \\
                      Straight driving, 3 cars & $\qty{100}{\percent}$& $\qty{97.6}{\percent}$ \\
                      Straight driving, 4 cars & $\qty{100}{\percent}$& $\qty{97.0}{\percent}$ \\
                      Merging, 3 cars & $\qty{100}{\percent}$ & \makecell{$\qty{72.8}{\percent}(\epsilon=\num{1e-1})$,\\$\qty{0}{\percent} (\epsilon=\num{1e-3})$} \\
                      \bottomrule
                  \end{tabular}
                  \label{tab:success_rate}
              \end{center}
              \vspace{\TableMargin}
          \end{table}
          
          \begin{figure}[htbp]
              \vspace{\SubFigureMargin}
              \centering
              \subfloat[]{
                  \includegraphics[width = 0.4\linewidth]{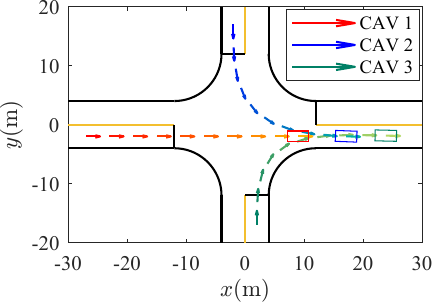}
                  \label{fig:svep_trajectory_i1s1l1rm}
              }
              \hfill
              \subfloat[]{
                  \includegraphics[width = 0.4\linewidth]{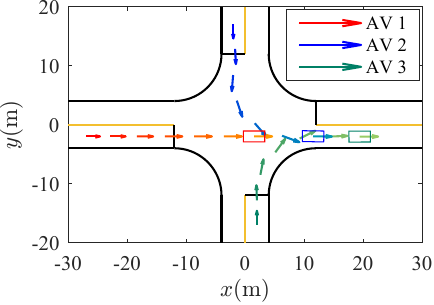}
                  \label{fig:algames_trajectory_i1s1l1rm}
              }
              \vspace{-0.3cm}
              \newline
              \subfloat[]{
                  \includegraphics[width = 0.4\linewidth]{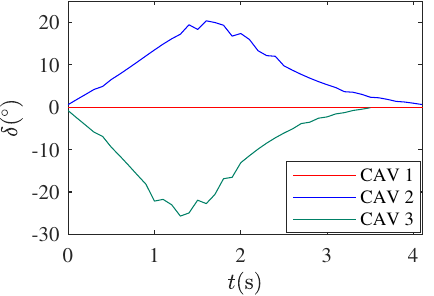}
                  \label{fig:svep_delta_i1s1l1rm}
              }
              \hfill
              \subfloat[]{
                  \includegraphics[width = 0.4\linewidth]{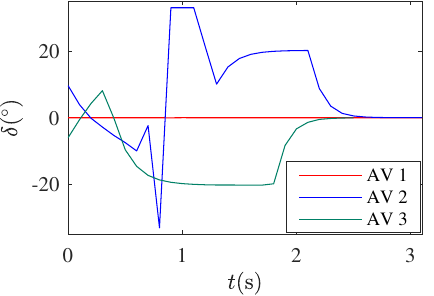}
                  \label{fig:algames_delta_i1s1l1rm}
              }
              \vspace{-0.2cm}
              \caption{Planning results using SVEP and ALGAMES. The arrow sequence shows the position and yaw angle of the vehicles every $\qty{0.4}{s}$, with its length proportional to the velocity. (a) SVEP's trajectory. (b) ALGAMES's trajectory. (c) SVEP's front wheel steering angles. (d) ALGAMES's front wheel steering angles.}
              \vspace{\FigureMarginBottom}
              \label{fig:trajectory_i1s1l1rm}
          \end{figure}

          Another important factor affecting safety is whether the equilibrium solved by each vehicle is the same, that is, whether each vehicle's predicted decisions of other vehicles match the actual decisions. We set the criterion for equilibrium concordance as the distance between predicted and actual $u_i(1)$ being less than $\num{1e-1}$, and define the average ratio of the number of times the planning between any pair of vehicles reaches equilibrium concordance in each experiment as the equilibrium concordance rate. The results are shown in Table \ref{tab:equilibrium_concordance_rate}, where for ALGAMES, we separately calculate the equilibrium concordance rate for successful cases (without car collision) and failed cases (with car collision).
          Through sharing $s_i$ and coordinating $\lambda_{i}$ by our designed Algorithm \ref{al:framework} (SVEP),  CAVs can obtain a consistent VE. In contrast, the uncoordinated ALGAMES may lead each vehicle to converge to different equilibrium. In the same scenario, the equilibrium concordance rate is higher in successful cases than in failed cases, which means that the concordant equilibrium can promote the safety of trajectories. Therefore, the V2X-based equilibrium coordination approach  in our Algorithm \ref{al:framework} can indeed  positively contribute to safety.
          
          It should be noted that although SVEP has performance advantages in experiments, these advantages come at the cost of relying on network communication.

          \begin{table}[htbp]
              \caption{The equilibrium concordance rate.}
              \vspace{\TableToCaption}
              \begin{center}
                  \begin{threeparttable}
                      \begin{tabular}{ccc}
                          \toprule
                          \multirow{2}{*}[-0.5\dimexpr \aboverulesep + \belowrulesep + \cmidrulewidth]{Scenario} & \multicolumn{2}{c}{Algorithm} \\ 
                          \cmidrule{2-3}
                          & SVEP & ALGAMES\tnote{1} \\ 
                          \midrule
                          Straight driving, 2 cars & $\qty{100}{\percent}$& $\qty{98.7}{\percent}$(S),$\qty{90.1}{\percent}$(F)\\
                          Straight driving, 3 cars & $\qty{100}{\percent}$& $\qty{99.3}{\percent}$(S),$\qty{95.7}{\percent}$(F)\\
                          Straight driving, 4 cars & $\qty{100}{\percent}$& $\qty{99.1}{\percent}$(S),$\qty{95.7}{\percent}$(F)\\
                          Merging, 3 cars   & $\qty{100}{\percent}$& $\qty{90.0}{\percent}$(S),$\qty{71.9}{\percent}$(F) \\
                          \bottomrule
                      \end{tabular}
                      \begin{tablenotes}
                          \item[1] Data is measured in successful (S) and failed (F) cases, respectively.
                      \end{tablenotes}
                  \end{threeparttable}
                  \label{tab:equilibrium_concordance_rate}
              \end{center}
              \vspace{\TableWithFootnoteMargin}
          \end{table}

          \label{it:last_experiment}
\end{enumerate}

\section{Conclusion}
\label{sec:conclusion}
We propose a semi-decentralized and VE-based CAV trajectory planner to resolve the computational efficiency and safety problems of previous uncoordinated game theoretic planners. By decomposing the computation process for each vehicle, the proposed planner can eliminate redundant computations to enhance efficiency. In addition, the multiplier consensus mechanism, with the help of the RSU, enables convergence to concordant equilibrium, thus enhancing safety and obtaining interaction-fair trajectories. Experimental results show that the proposed method meets the real-time implementation requirements, has favorable scalability, and ensures safety.


\bibliographystyle{IEEEtran}
\bibliography{reference}

\end{document}